\documentclass[preprint,showkeys,preprintnumbers,amsmath,amssymb]{revtex4}

\usepackage{amsthm}

\begin{document}

\title{Boundary driven zero-range processes in random media}

\author{Otto Pulkkinen}

\email{otto.pulkkinen@jyu.fi}

\affiliation{Department of Physics, P.O.~Box 35, FI-40014, University of Jyv\"askyl\"a, Finland}

\begin{abstract}
The stationary states of boundary driven zero-range processes in random media with quenched disorder are examined, 
and the motion of a tagged particle is analyzed. For symmetric transition rates, also known as the random barrier 
model, the stationary state is found to be trivial in absence of boundary drive. Out of equilibrium, two further 
cases are distinguished according to the tail of the disorder distribution. For strong disorder, the fugacity 
profiles are found to be governed by the paths of normalized $\alpha$-stable subordinators. The expectations of 
integrated functions of the tagged particle position are calculated for three types of routes.   
\end{abstract}


\maketitle

\section{\label{SecIntro}Introduction}

The zero-range process, first introduced by Spitzer \cite{Spitzer70}, is a simple model for a gas of mutually 
interacting particles on a finite or infinite lattice. The interaction is local in the most strict sense: the rate, 
at which particles leave a lattice site, depends on the number of particles present on that particular site only. 
This restriction makes the stationary state analytically tractable--it is a product measure--even in presence of 
driving fields that break the spatial symmetry of the transition rates. The invariant measures of the process on 
${\mathbb Z}^d$ were studied by Andjel \cite{Andjel82}. We shall consider finite lattices only.

Due to its analytical tractability, the zero-range process has been used to test ideas concerning the stationary 
states and dynamics of real interacting systems. The hydrodynamic limits were given a rigorous treatment in the book 
of Kipnis and Landim \cite{KipnisLandim99}, which provided a basis for recent studies concerning the thermodynamic 
functionals of nonequilibrium steady states \cite{Bertini02}. The applications of zero-range processes in physics 
range from shaken granular gases to condensation phenomena (see \cite{EvansHanney05} and references therein).    
	
The stationary states of zero-range processes with open boundaries have already been discussed in detail by Levine, 
Mukamel, and Sch\"utz \cite{Levine05}. In this article, we shall consider open systems but the media itself will be a 
random in the sense that the time evolution of the particle system takes place in a fixed environment, but the exact 
structure of the media is not known--only its statistics for a generic sample. This type of randomness is known as 
quenched disorder. Bulk driven zero-range processes in random media have been used to study jamming transitions in 
asymmetric exclusion processes with particle-wise disorder \cite{Evans96,Krug96}. We shall show that there is plenty 
to discover even in the boundary driven but otherwise symmetric system. In particular, the stationary state shows 
nontrivial correlations when the system is driven out of equilibrium. By restricting the dynamics to the most simple 
and natural case of symmetric media, where the jump rate from lattice site $i$ to $i+1$ equals the rate in the 
opposite direction (imagine the particles surmounting barriers with independent and identically distributed random 
heights), we shall find that, due to universal laws for normalized sums of independent and identically distributed 
random variables, closed expressions for a number of quantities can be found. 

In addition to characterizing the structure of the stationary states in random media, we shall also focus on the 
motion of a single, tagged particle in a gas of zero-range interacting particles. This problem is on one hand 
connected to the theory of diffusions and random walks in random environments reviewed in 
\cite{HavlinBenAvraham87,HughesVol2,Zeitouni06} (Isichenko \cite{Isichenko92} gives emphasis to the motion of a 
passive tracer in random flows). Most of the studies concerning random walks in random environments allow the 
asymmetry of the transition probabilities, which makes the problem much harder for zero-range processes. Even the 
model with particles diffusing in a landscape of valleys with random depths (the particle jumps out of a site to 
either direction with equal probability \cite{HavlinBenAvraham87}) presents some extra difficulty as compared to our 
model with barriers of random height, because the existence of the stationary state is, in that case, not easily 
controlled. Rigorous results for a single particle in symmetric random environment were derived by Kawazu and Kesten 
\cite{Kawazu84}. They proved that the suitably speeded up random walk converges on the level of path measures to an 
$\alpha$-stable diffusion as the process is observed on larger and larger length scales. Such convergence results are 
sometimes called "invariance principles" \cite{Ferrari96}.    
 
On the other hand, there exists a vast literature on limits of tagged particle motion in exclusion and zero-range 
processes, and some other interacting particle systems \cite{Ferrari96}, mostly in homogeneous media. There are a few 
important contributions for zero-range processes: For symmetric and translation invariant transition probabilities, 
Saada \cite{Saada90} proved a central limit theorem for the position of the tagged particle for a particular 
interaction, and Siri \cite{Siri98} showed an invariance principle for processes on tori, with the interaction 
possibly depending on the location. Sethuraman \cite{Sethuraman06} considered the variance of the particle position 
in asymmetric cases. 

We shall show how the expectations of certain functionals of the paths of the tagged particle, more precisely 
functions of the tagged particle position integrated up to stopping times of an auxiliary process, can be calculated 
first in the finite system and, in the end, as the macroscopic limit is approached. The convergence of the motion in 
random symmetric media to a diffusion on a formal level is not proved.
  
The structure of the article is as follows. In the next section, we define the model and solve its stationary 
distribution in a fixed environment. For later use in this article, we also derive the generator of the adjoint 
dynamics. Section \ref{SecSrm} discusses the stationary states in random media. The cases of weakly 
(\ref{SecStatLessThanInfty}) and strongly (\ref{SecStatInfty}) inhomogeneous environments are treated separately. 
Some properties of the fugacity profiles in the latter case are discussed in section \ref{SecDistribution}. The 
section \ref{SecTagged} is devoted to the motion of a tagged particle, again first in a given environment (section 
\ref{SecTravelsFinite}). Its movement in random media is considered in \ref{SecTravelRandom}. The results are 
commented and related to other studies in the last section. 

Throughout the article, the symbols $a \vee b$ and $a \wedge b$ are used, respectively, to denote the maximum and the 
minimum of two numbers $a$ and $b$, and $\lfloor x \rfloor$ stands for the greatest integer smaller than or equal to 
$x$.    

\section{Stationary distribution}
\protect\label{SecStat} 

\subsection{Preliminaries}
\label{SecPreli}

The nearest neighbour zero-range process $(X(t))_{t\geq 0}=(X_1 (t),\ldots ,X_n (t))_{t\geq 0}$ on $[n]=\lbrace 
1,2,\ldots,n \rbrace$ with left (resp.\ right) boundary reservoir at $0$ (resp.\ $n+1$) is a Markov process with 
values in ${\mathbb Z}_+^n$ for each $t\geq 0$, and generated by
{\setlength\arraycolsep{2pt}
\begin{eqnarray}
\label{Lnf}
L_n f (\eta) &=& v_0 \lbrack f(\eta^{0,1}) - f(\eta) \rbrack+  u_{n+1} \lbrack f(\eta^{n+1,n}) - f(\eta) \rbrack 
\nonumber \\ 
& & + \sum_{i=1}^{n} g(\eta_i) \lbrace v_i  \lbrack f(\eta^{i,i+1}) - f(\eta) \rbrack +  u_{i} \lbrack 
f(\eta^{i,i-1}) - f(\eta) \rbrack \rbrace,
\end{eqnarray}}
where $\eta^{i,j}$ is the state obtained from $\eta$ by moving a particle from $i$ to $j$ with 
the convention that $\eta^{0,1}=(\eta_1 +1, \eta_2, \ldots, \eta_n)$ etc.\ for transitions involving boundary 
reservoirs. The single particle rates $(u_i,v_i)$ for the jumps to the left and right, respectively, are assumed 
positive and bounded, and the interaction $g:\mathbb{Z}_+ \longrightarrow \mathbb{R}_+$ is such that, for some 
$\epsilon>0$, $0=g(0)<\epsilon \leq g(k)$ for $k > 0$, and has bounded increments;
$\sup_{k} \vert g(k+1)- g(k) \vert \leq K$.

The next proposition involves the mean current from $j$ to $j-1$, which is defined as the expected number of 
particles that move from $j$ to $j-1$ during one unit of time minus the expected number of particles that move in the 
opposite direction. We shall also use the shorthand notation $g!(m)=\prod_{j=1}^{m} g(j)$ for $m>0$ and $g!(0)=1$.
\theoremstyle{definition}
\newtheorem{prop}{Proposition}
\begin{prop}
The stationary distribution of the zero-range process generated by (\ref{Lnf}) is the product measure $\mu (\eta) = 
\prod_{j=1}^{n} \nu_{\phi_j}(\eta_j)$ with 
{\setlength\arraycolsep{2pt}
\begin{eqnarray}
\label{statgeneral1} 
\nu_{\phi}(m) &=& \frac{1}{Z(\phi)} \frac{\phi^{m}}{g!(m)},\ \ Z(\phi)= \sum_{m=0}^{\infty} \frac{\phi^{m}}{g!(m)} \\
\label{statgeneralphi}
\phi_j &=& \frac{v_0}{v_j} \prod_{i=1}^{j} \frac{v_i}{u_i} + \frac{c}{v_j} \sum_{i=0}^{j-1} \prod_{k=0}^{i} 
\frac{v_{j-k}}{u_{j-k}}, \\
\label{statgeneralc}
&& c = \frac{u_{n+1}- v_0 \prod_{i=1}^{n} \frac{v_i}{u_i} }{1+ \sum_{i=0}^{n-1} \prod_{k=0}^{i} 
\frac{v_{n-k}}{u_{n-k}}} ,
\end{eqnarray}}
given that the {\it fugacities} $\phi_j$ satisfy $\phi_j< \liminf_{m\to \infty} g(m)$ for $j\in [n]$. Here $c$ equals 
the mean current from $j$ to $j-1$ for any pair of sites with $j\in [n+1]$.
\end{prop}
\begin{proof}
We show that $E^{\mu} \lbrack (L_n f)(X(t)) \rbrack = \mu \lbrack L_n f \rbrack =0$ for bounded $f$. By a change of 
summation index, this holds if $\mu$ satisfies the balance equations 
{\setlength\arraycolsep{2pt}
\begin{eqnarray}
&& \lbrack u_1 g(\eta_1 +1) \mu (\eta^{0,1}) - v_0 \mu (\eta) \rbrack +  
\lbrack v_n g(\eta_n +1) \mu (\eta^{n+1,n}) - u_{n+1} \mu (\eta) \rbrack \\
&& {}+ \sum_{j=1}^n \,\lbrack v_{j-1} g(\eta_{j-1}+1) \mu (\eta^{j,j-1}) + u_{j+1} g(\eta_{j+1}+1) \mu (\eta^{j,j+1}) 
- (v_j + u_j) g(\eta_j) \mu(\eta) \rbrack = 0 \nonumber 
\end{eqnarray}}
for $\eta \in \mathbb{Z}_+^n$. Substitution of the product form with (\ref{statgeneral1}) and some simple algebra 
leads to traffic equations for the fugacities $\phi_j$:
{\setlength\arraycolsep{2pt}
\begin{eqnarray}
\label{trafficgeneral}
u_{j+1} \phi_{j+1} - v_j \phi_j &=& u_j \phi_j - v_{j-1} \phi_{j-1} \\
u_1\phi_1 - v_0 &=& u_{n+1} - v_n \phi_n  ,
\end{eqnarray}}
the solution to which is given by (\ref{statgeneralphi}) and (\ref{statgeneralc}) with $c$ equal to the quantities on 
both sides of the equations. The expected number of particles that leave site $j$ in a time interval of unit length 
is $\sum_{\eta} (u_j+v_j )g(\eta_j) \mu(\eta) = (u_j+v_j) \phi_j$. So $c$ really is the mean current. 
\end{proof}

Given that $\phi_j< \liminf_{m\to \infty} g(m)$ for $j\in [n]$, the expected number of particles at site $j$ is 
finite and given by the expectation of a $\nu_{\phi_j}$-distributed random variable,
\begin{equation}
\rho_j := R(\phi_j) := \phi_j Z'(\phi_j)/Z(\phi_j).
\end{equation}
Notice that the function $R$ is strictly increasing because $\phi R'(\phi)$ is the variance of the distribution 
$\nu_{\phi}$, which is positive by the assumption that the increments of the interaction $g$ are bounded.  

Formula (\ref{statgeneralc}) shows that, in general, there is a current of particles due to the imbalance of the 
fields $(u_i)$ and $(v_i)$, which means that the process is asymmetric with respect to reversal of time. Since the 
time reversed process is generated by the adjoint of $L_n$ in $L^2 (\mu)$, that is $\mu \lbrack h (L_n f) \rbrack = 
\mu \lbrack (L_n^{\ast} h)  f \rbrack$, change-of-variable calculations lead to the following:  
\begin{prop}
The time reversal of $(X(t))$ with respect to $\mu$ is a zero-range process with the generator 
{\setlength\arraycolsep{2pt}
\begin{eqnarray}
\label{Lnastf}
L_n^{\ast} f (\eta) &=& u_1 \phi_1 \lbrack f(\eta^{0,1}) - f(\eta) \rbrack +  v_{n} \phi_n \lbrack f(\eta^{n+1,n}) - 
f(\eta) \rbrack  \nonumber \\ 
& & + \sum_{i=1}^{n} g(\eta_i) \Big\lbrace u_{i+1} \frac{\phi_{i+1}}{\phi_i}  \lbrack f(\eta^{i,i+1}) - f(\eta) 
\rbrack +  v_{i-1} \frac{\phi_{i-1}}{\phi_i} \lbrack f(\eta^{i,i-1}) - f(\eta) \rbrack \Big\rbrace,
\end{eqnarray}}
i.e.\ it has single particle rates $v_i^{\ast}=u_{i+1}\phi_{i+1} /\phi_i$ and $u_i^{\ast}=v_{i-1}\phi_{i-1} /\phi_i$
with the convention that $\phi_0= \phi_{n+1}=1$.
\end{prop}

This shows that the process is reversible if and only if the detailed balance conditions 
$v_i\phi_i=u_{i+1}\phi_{i+1}$ hold, i.e. $c=0$.

In this article, we shall restrict our attention to symmetric media. This means $u_{j+1} = v_j$ for $j\in [n-1]$
and possible exceptions on the boundaries, where we take $v_{0} \equiv v$ as the input rate to the system from the 
left reservoir and $u_{n+1} \equiv u$ from the right. The output rates at the boundaries are chosen to equal unity, 
$u_1= v_n =1$. Define 
\begin{equation}
\label{s_j}
s_j = 1+ \sum_{i=1}^{j-1} v_i^{-1},
\end{equation}
with $s_1=1$. Then proposition $1$ yields the following:
\newtheorem{coro}{Corollary} 
\begin{coro}
For symmetric media with boundary fields $u,v$, the stationary distribution is of the product form as in Proposition 
$1$ with
{\setlength\arraycolsep{2pt}
\begin{eqnarray}
\label{fuga_i}
\phi_i &=& u \frac{s_i}{s_{n+1}} + v \Big(1-\frac{s_i}{s_{n+1}} \Big), \\
&&\ \ c = \frac{u-v}{s_{n+1}},
\end{eqnarray}}
given that $\phi_1 \vee \phi_n < \liminf g(k)$.
\end{coro}
Naturally, the process is time reversible if and only if the boundary rates coincide, and then the fugacity takes a 
constant value equal to the input rates. For general $u$ and $v$, $\phi_i$ is a monotone function of the spatial 
position, and therefore a condition only at its boundary values need be imposed to guarantee the existence of the 
stationary distribution. In the following sections, we shall adopt a stricter condition, $u \vee v < \liminf g(k)$, 
than the one used in the corollary.  

\subsection{Symmetric random media}
\label{SecSrm}

\newcommand{\Prm}{\mathbb{P}}
\newcommand{\Erm}{\mathbb{E}}
\newcommand{\ud}{\mathrm{d}}
We shall now consider zero-range processes in symmetric random environments with quenched disorder. This means that 
the environment does not evolve in time. For this purpose, let the single particle rates $(v_i)_{i=1}^{n-1}$ for 
jumps between $i$ and $i+1$ in the bulk be independent and identically distributed random variables on some 
probability space with measure $\Prm$. Expectations with respect to this measure will be denoted by $\Erm$. The 
symbols $P$ and $E$ are reserved for the Markov process in a fixed environment. Again, the assumption of positive and 
bounded rates is made. Each random environment is equipped with the same boundary rates as in the setting of 
corollary 1: The input rate $v_0$ from the left reservoir to site $1$ is denoted more conveniently by $v$, the rate 
$u_{n+1}$ from the right reservoir to site $n$ by $u$, and the output rates from sites $1$ and $n$ to the left and 
right reservoirs, that is $u_1$ and $v_n$ respectively, are set to unity. The sufficient condition 
\begin{equation}
\label{uvcond}
u \vee v < \liminf g(k)
\end{equation} 
for the existence of the stationary state is always assumed, and unless stated otherwise, the equilibrium case $u=v$ 
is excluded from considerations. 

\subsubsection{Weak disorder: ${\Erm} v_i^{-1} < \infty$}
\label{SecStatLessThanInfty}

The existence of expectations of the inverse bulk rates determines the large scale structure of the system under 
study. In fact, the situation is rather simple for ${\Erm} v_i^{-1} < \infty$: Defining a rescaled continuum analogue 
of Eqn.\ (\ref{s_j}) by   
\begin{equation}
\label{S_nDef}
S_{n,\alpha} (x) := n^{-1/\alpha} \Big(\, 1 +\!\!\! \sum_{i=1}^{\lfloor (n+1)x \rfloor -1}\!\! v_i^{-1} \Big), \qquad 
x\in \lbrack 0, 1\rbrack,
\end{equation}
and choosing $\alpha=1$ yields that $S_{n,1} (x) \to x\Erm v_i^{-1}$ for $\Prm$-almost all environments as $n\to 
\infty$. The convergence is even uniform because the functions $S_{n,\alpha}(x)$ are increasing, and the limiting 
function is continuous and increasing. Thus the continuum version of the fugacity satisfies
\begin{equation}
\phi_n(x):= u \frac{S_{n,1} (x)}{S_{n,1} (1)} + v \Big(1-\frac{S_{n,1} (x)}{S_{n,1} (1)} \Big) \longrightarrow 
ux+v(1-x)
\end{equation}
uniformly almost surely.
Moreover, the current is inversely proportional to the system size,
\begin{equation}
n c =  \frac{u-v}{S_{n,1} (1)} \longrightarrow \frac{u-v}{{\Erm} v_i^{-1}}.
\end{equation}
The continuity of the map $R(\phi) = \phi Z'(\phi)/Z(\phi)$ on $[u\wedge v, u\vee v]$ implies the uniform convergence 
of the particle density,
\begin{equation}
\rho_n (x) := (R\circ \phi_n)(x) \longrightarrow  R(ux+v(1-x))
\end{equation}
$\Prm$-almost surely.
As shown in \cite{Bertini02}, the fluctuations of the macroscopic density profile are described by the free energy of 
an equilibrium system with spatially varying mean waiting times $\Erm (v_i + u_i)^{-1}$ used to create the same 
fugacity profile. This should be contrasted with more complicated non-equilibrium systems, such as exclusion 
processes, with non-local large deviation functions \cite{Derrida02}.

\subsubsection{Strong disorder: $\Erm v_i^{-1} = \infty$}
\label{SecStatInfty}

We now turn to the nontrivial case of $\Erm v_i^{-1} = \infty$. In particular, we assume that $(v_i^{-1})$ belong to 
the domain of attraction of a strictly $\alpha$-stable distribution \cite{Sato99,Ibragimov71}, with $\alpha \in 
(0,1)$, and 'strictly' meaning absence of a deterministic component as in Sato \cite{Sato99}. In other words, there 
is a slowly varying function $L$ such that $(n L(n))^{-1/\alpha}\sum_{i=1}^n v_i^{-1}$ converges in distribution to a 
strictly $\alpha$-stable random variable $S$ as $n\to \infty$, or equivalently (see \cite{Ibragimov71} for a proof), 
the tail of the distribution of the $v_i^{-1}$ is heavy with exponent $\alpha$. 
For simplicity, we take $L(n)=1$. The sums of $v_i^{-1}$ being almost surely positive, the properly rescaled current 
converges in distribution by theorem 5.2 of reference \cite{Billingsley68}, 
\begin{equation}
n^{1/\alpha} c \longrightarrow (u-v)/S.
\end{equation}
There is also a limit theorem for the fugacity as a stochastic process. The processes $\lbrace S_{n,\alpha} (x), 
x\in[0,1] \rbrace$, defined by equation (\ref{S_nDef}), as elements of the space $D\lbrack 0,1\rbrack$ of 
right-continuous functions on $\lbrack 0,1\rbrack$ with left-hand limits, converge in the Skorohod topology (implied 
by the metric $d(\phi,\psi) = \inf\lbrace \sup_x \vert \phi(x)- \psi(\lambda(x))\vert \vee \sup_x \vert x- 
\lambda(x)\vert : \lambda \ {\textrm{continuous and one-to-one}} \rbrace$ 
\cite{Skorohod56,Skorohod57,Gihman74,Billingsley68}) to a strictly $\alpha$-stable subordinator (increasing L\'evy 
process \cite{Sato99,Bertoin96}) $\lbrace S(x), x\in[0,1] \rbrace$. 
The set of discontinuities of the map $\lbrace S_{n,\alpha} (x), x\in[0,1] \rbrace \mapsto \lbrace S_{n,\alpha} 
(x)/S_{n,\alpha} (1), x\in[0,1] \rbrace$ from the space $D\lbrack 0,1\rbrack$ to itself consists, in this topology, 
only of the path that is identically zero, and is thus of null measure for the process $S$. Then, by theorem 5.1 of 
reference \cite{Billingsley68}, the fugacity processes $\phi_n(x)$ converge in distribution to
\begin{equation}
\phi(x) = u \frac{S(x)}{S(1)} + v \Big(1-\frac{S(x)}{S(1)} \Big),
\end{equation}
where the normalized strictly $\alpha$-stable subordinator $S(x)/S(1)$ on $\lbrack 0,1 \rbrack$ is also a random 
distribution function. This object is rather well known in mathematical statistics as a part of a family generalizing 
the Dirichlet process \cite{Pitman02Notes}. The next theorem gives the one-point statistics of the fugacity.
\newtheorem{thm}{Theorem}
\begin{thm}
The distribution function of the fugacity $\phi$ at $x\in (0,1)$ is 
\begin{equation}
\label{phidf}
\Prm (\phi(x)\leq \theta) = \frac{1}{2} - \frac{\textrm{sgn}(u-v)}{\pi \alpha} \arctan \bigg( \frac{x\vert u-\theta 
\vert^{\alpha} -(1-x)\vert \theta - v \vert^{\alpha}}{x\vert u-\theta \vert^{\alpha} + (1-x)\vert \theta - v 
\vert^{\alpha}} \tan  \frac{\pi\alpha}{2} \bigg)
\end{equation}
and its moments can be calculated from the formula 
\begin{equation}
\Erm\phi(x)^n =v^n + \sum_{k=1}^n {{n}\choose{k}} v^{n-k} (u-v)^k \sum_{j=1}^k x^j q_{\alpha,0}(k,j),
\end{equation}
where  
\begin{equation}
\label{qalpha0}
q_{\alpha,0}(k,j) = k (j-1)! \alpha^{j-1} \sum \prod_{i=1}^k \frac{((1-\alpha)(2-\alpha)\cdots 
(i-1-\alpha))^{m_i}}{m_i!(i!)^{m_i}}.
\end{equation}
The summation is over the sequences of nonnegative integers $(m_i)_{i=1}^k$ such that $\sum_{i=1}^k m_i =j$ and 
$\sum_{i=1}^k i m_i =k$.
\end{thm}
\begin{proof}
Since the normalized subordinator $S(x)/S(1)$ is a random distribution function on the set of real numbers, there is 
a random probability measure, to be denoted by $\psi$, associated with it. Regazzini, Lijoi and Pr\"unster 
\cite{Regazzini03} have derived expressions for distribution functions of integrals with respect to such measures 
using the Gurland inversion formula \cite{Gurland48}
\begin{equation}
F(\xi ) +F(\xi -) = 1 - \frac{1}{i \pi} \lim_{\substack{\epsilon\to 0 \\ T \to \infty}} \int_{\epsilon}^T 
\frac{e^{i t \xi}}{t}\, \varphi (t)\, \ud t
\end{equation}
for a distribution function $F$ and the corresponding charateristic function $\varphi$. Since for $u>v$,
\begin{equation}
\phi(x) \leq \theta \Leftrightarrow \frac{S(x)}{S(1)} \leq \frac{\theta-v}{u-v} 
\Leftrightarrow \psi \lbrack I\lbrace \cdot \leq x \rbrace \rbrack \leq \frac{\theta-v}{u-v},
\end{equation}
choosing $f=I\lbrace \cdot \leq x \rbrace$ in formula (10) on page 575 of reference \cite{Regazzini03}, which in our 
case reads
{\setlength\arraycolsep{2pt}
\begin{eqnarray}
\frac{1}{2} \lbrack \Prm (\psi[f] < \theta) &+& \Prm (\psi[f] \leq \theta) \rbrack= \frac{1}{2} \nonumber\\
&-& \frac{1}{\pi \alpha} 
\arctan \left( \frac{\int_0^1 {\mathrm{sgn}} (f(y)-\theta ) \vert f(y)-\theta \vert^{\alpha} \ud y}{\int_0^1 \vert 
f(y)-\theta \vert^{\alpha} \ud y} \tan \frac{\pi \alpha}{2} \right),
\end{eqnarray}}
and integrating yields the first statement of the theorem. Notice that there is null probability that $\phi(x)$ takes 
a fixed value for $x\in (0,1)$, and thus $\Prm (\phi(x)< \theta) = \Prm (\phi(x) \leq \theta)$.

To prove the second statement, by the binomial theorem we need to show that the $k$-th moment of the normalized 
subordinator at $x$ equals the generating function for the sequence $q_{\alpha,0}$(k,j), which is in fact a 
probability distribution, as will be explained in remarks after the proof. The use of double subscript shall also be 
justified. The statement is verified by the following calculation based on conditioning, Fubini theorem and Fa\`a di 
Bruno formula. Here $(\mathcal{F}_x)_{x\geq 0}$ is the natural filtration of the subordinator $S(x)$, and 
$\Phi(\lambda)=\lambda^{\alpha}\Phi(1)$ is the Laplace exponent of its (infinitely divisible) distribution at $x=1$.  
{\setlength\arraycolsep{2pt}
\begin{eqnarray}
\lefteqn{ \mathbb{E} \lbrack S(x)^k S(1)^{-k} \rbrack = 
\Erm \Big\lbrack S(x)^k \Erm \Big\lbrack \int_0^{\infty}\!\!\! \cdots  \int_0^{\infty} e^{-S(1) \Sigma_{1}^k
 \lambda_i} 
\,\ud \lambda_1 \cdots \ud \lambda_k \Big\vert \mathcal{F}_x \Big\rbrack \Big\rbrack } \nonumber\\
&=& \Erm \Big\lbrack  \int_0^{\infty}\!\!\! \cdots  \int_0^{\infty} S(x)^k  e^{-S(x)\Sigma_{1}^k \lambda_i- (1-x) 
\Phi(\Sigma_{1}^k \lambda_i)} \,\ud \lambda_1 \cdots \ud \lambda_k  \Big\rbrack \nonumber\\
&=& (-1)^k \int_0^{\infty}\!\!\! \cdots  \int_0^{\infty} \Big(\frac{\partial^k}{\partial\lambda_1 \cdots 
\partial\lambda_k} e^{-x\Phi(\Sigma_{1}^k \lambda_i)} \Big)\,  e^{-(1-x)\Phi(\Sigma_{1}^k \lambda_i)}\,\ud \lambda_1 
\cdots \ud \lambda_k \nonumber\\
&=& (-1)^k \int_0^{\infty} \frac{r^{k-1}}{(k-1)!} \Big(\frac{\ud^k}{\ud r^k} e^{-x \Phi(r)} \Big)\,  
e^{-(1-x)\Phi(r)}\,\ud r \nonumber\\
&=&  \frac{(-1)^k}{(k-1)!} \int_0^{\infty} r^{-1} e^{-r^{\alpha}\Phi(1)} \sum {{k}\choose{m_1 \cdots m_k}} (-x r^{ 
\alpha}\Phi(1))^{\Sigma_1^k m_i} \prod_{i=1}^k \Big(\frac{(\alpha)_{i \downarrow} }{i!}  \Big)^{m_i} \, \ud r 
\nonumber\\
&=& \sum x^{\Sigma_1^k m_i} k (-1+\sum_1^k m_i)! \alpha^{(-1+\Sigma_1^k m_i)} \prod_{i=1}^k 
\frac{((1-\alpha)(2-\alpha)\cdots (i-1-\alpha))^{m_i}}{m_i!(i!)^{m_i}} \nonumber\\
&=& \sum_{j=1}^k x^j q_{\alpha,0} (k,j), \nonumber
\end{eqnarray}}
where the summation without indices on the two lines before the last one is over the set $\lbrace (m_i)_{i=1}^k  : 
m_i\geq 0,\ \sum_{i=1}^k i m_i = k \rbrace$ and $(\alpha)_{i \downarrow} = \alpha (\alpha-1)\cdots (\alpha-i+1)$.
\end{proof}

The double index notation in $q_{\alpha,0}$ comes from the theory of exchangeable random partitions, and is taken 
from J. Pitman's lecture notes \cite{Pitman02Notes}. The list of research articles on this topic include 
\cite{Perman92,Pitman96,Pitman97,PitmanReport625}. We quickly review here the meaning of the sequence 
$q_{\alpha,0}(k,j)$ in the context of partitions.

An exchangeable random partition is a partition of the set of positive integers $\mathbb{N}$ generated by sampling 
from a random discrete probability distribution. A partition is identified with the sequence of partitions of 
truncated sets $\lbrack n \rbrack = \lbrace 1, \ldots, n\rbrace$ in such a way that the probability of a partition of 
$\lbrack n \rbrack$ in $k$ parts with cardinalities $(n_1,\ldots, n_k)$ is given by the symmetric function 
\cite{Pitman96} $p(n_1,\ldots, n_k) = \sum_{(j_1,\ldots,j_k )} \Erm  \prod_{i=1}^k P_{j_i}^{n_i}$. 
Here the summation is over all sequences of $k$ distinct positive integers, and $P_i$ is the $i$-th largest atom of 
the random sampling distribution. One way of thinking here is through the species sampling models of ecology: A 
sample of $n$ individuals from a large population, with ranked relative abundances given by the sequence $(P_i)$, is 
partitioned according to the species in that sample. One possibility is to take $(P_i)$ from the statistics of jumps 
of an $\alpha$-stable subordinator, which provides the connection to the system under study. A more general scheme 
known as the two-parameter family of Poisson--Dirichlet distributions, and denoted by $\mathrm{PD}(\alpha, \theta )$, 
is presented in \cite{Pitman97}. The number $q_{\alpha,\theta}(k,j)$ is then the probability that in a sample of size 
$k$ from a population with frequencies $\mathrm{PD}(\alpha, \theta )$, there are exactly $j$ species. 

Poisson--Dirichlet distributions appear naturally in the theory of spin-glasses \cite{Talagrand03}. In particular, 
the ranked Gibbs weights have this distribution in the thermodynamic limit of Derrida's random energy model 
\cite{Derrida81}. $\mathrm{PD}(1/2,0)$ also appears as the statistics of ranked lengths of Brownian excursions, and 
$\mathrm{PD}(\alpha,0)$ for excursions of more general objects such as Bessel processes. Therefore one expects a wide 
range of applications for these structures in the physics of random systems.

\theoremstyle{remark} \newtheorem*{example}{Example}
\begin{example}
As a specific application of theorem 1, we now determine the density profile near the boundaries for the attractive 
interaction $g(k)=I\lbrace k>0 \rbrace$. For $u,v<1$, we get
$R(\phi(x)) = \phi(x)/(1-\phi(x)) = \sum_{k=1}^{\infty} \phi (x)^k$.
By theorem 1, and the fact that the expected number of species in a sample of size $k$ with $\mathrm{PD}(\alpha,0)$ 
frequencies is $(\alpha+1)(\alpha+2)\cdots  (\alpha+k-1) / (k-1)!$ \cite{Pitman02Notes}, 
{\setlength\arraycolsep{2pt}
\begin{eqnarray}
\lim_{x\to 1} \frac{\ud}{\ud x} \Erm  R (\phi (x)) &=& \frac{1}{1-v}\sum_{k=1}^{\infty} \left(\frac{u-v}{1-v} 
\right)^k \sum_{j=1}^k j q_{\alpha,0}(k,j) \nonumber\\
&=& \frac{u-v}{(1-u)(1-v)}\left(\frac{1-v}{1-u}\right)^{\alpha}. \nonumber 
\end{eqnarray}}
On the other hand, the derivative at the opposite boundary, $x\to 0$, is
\begin{equation}
\frac{1}{1-v}\sum_{k=1}^{\infty} \left(\frac{u-v}{1-v} \right)^k q_{\alpha,0}(k,1) = 
\frac{u-v}{(1-u)(1-v)}\left(\frac{1-u}{1-v}\right)^{\alpha} \nonumber 
\end{equation}
because $q_{\alpha,0}(k,1) = (1-\alpha)(2-\alpha)\cdots (k-1-\alpha)/(k-1)!$ by formula (\ref{qalpha0}) of theorem 1.  
\end{example}

\subsubsection{Properties of the fugacity distribution}
\label{SecDistribution}

Let us now discuss some consequences of theorem 1. First of all, the result for the moments shows that the expected 
value for the fugacity at $x$ is $\Erm \phi(x) = v(1-x)+ ux$. The formula $\Erm \phi(x)^2 - \left( \Erm \phi(x) 
\right)^2 = (1-\alpha)(u-v)^2 x(1-x)$ for the variance allows us to estimate the magnitude of deviations around the 
linear profile;
\begin{equation}
\Prm \Big( \vert \phi(x) - v(1-x)- ux \vert > \vert u-v \vert \sqrt{x(1-x)} \Big) \leq 1-\alpha.   
\end{equation}
This matches with the intuitive picture of the fugacity process: By the L\'evy-It\^o decomposition, the 
discontinuities are countable and dense on the spatial axis, but for $\alpha$ close to $1$, the jumps do not differ 
in size very much, and the profile is usually nearly linear. For $\alpha$ tending to zero, a small number of jumps 
dominates the normalized sum in the fugacity, and the system is segregated into regions of nearly constant particle 
densities separated by bonds of very low conductivity.

More information on the statistics of the fugacity can be extracted by differentiating the distribution function 
(\ref{phidf}). For convenience, we consider the case $v<u$ only, and introduce the notation $\lambda = \theta-v$, 
$\xi= u-\theta$. Then the probability density reads
\begin{equation}
\label{f_phi}
f_{\phi(x)}(\theta) = \frac{\sin \pi \alpha}{\pi} \frac{(u-v)x(1-x)\lambda^{\alpha-1} \xi^{\alpha-1}}
{( x \lambda^{\alpha} + (1-x) \xi^{\alpha} )^2 \cos^2\frac{\pi \alpha}{2}+  
( x \lambda^{\alpha} - (1-x) \xi^{\alpha} )^2 \sin^2\frac{\pi \alpha}{2}}.
\end{equation} 
It resembles the generalized arcsine distribution that appears as a solution to some occupation time problems. A 
common feature of the two densities is that they both diverge at the boundaries, which is perhaps unexpected for 
$\alpha$ close to $1$ in the present case. The main difference is that the arcsine density is convex, whereas the 
density given above has two minima for large enough $\alpha$. Taking into account the divergence at the boundaries 
and confinement within $ux+v(1-x) \pm (u-v)\sqrt{x(1-x)}$ as $\alpha \to 1$, this is no longer a surprise. 

It is easily seen that the number of extrema of the density (\ref{f_phi}) does not depend on the actual values of the 
boundary fields, as long as they are not equal. So without loss of generality, we choose $v=0$ and $u=1$. The number 
of extrema within $(0,1)$ is then given by the number of solutions to
\begin{equation}
(2\theta -1 -\alpha)\theta^{2\alpha}(1-x)^2 + (2\theta-1+\alpha)(1-\theta)^{2\alpha}x^2 + 2(2\theta 
-1)(1-\theta)^{\alpha}\theta^{\alpha}(1-x)x \cos \pi \alpha= 0.
\end{equation}
Let us concentrate for simplicity on the statistics in the middle of the system, i.e.\ $x=1/2$. Writing $\delta = 
\theta - 1/2$, we get
\begin{equation}
\label{extrema}
\left( \frac{1}{2}+ \delta \right)^{2\alpha} (2\delta -\alpha) + \left( \frac{1}{2} - \delta \right)^{2\alpha}  
(2\delta +\alpha) + 4\delta \left( \frac{1}{2} + \delta \right)^{\alpha}\left( \frac{1}{2} - \delta 
\right)^{\alpha}\cos \pi \alpha= 0,
\end{equation} 
which has a solution $\delta=0$ for every $\alpha\in (0,1)$. Moreover, the left-hand side is an odd function of 
$\delta$, so if nontrivial solutions exist, they are found symmetrically from both sides of the origin. Notice also 
that the derivative of the density is negative for $\delta$ just above $-1/2$, and positive for $\delta$ close to 
$1/2$. There are at most three solutions. Let $\alpha^{\ast}\approx 0.5946$ be the solution to $1 - 2\alpha^2 +\cos 
\pi \alpha =0$. Then the expansion of the left-hand side of equation (\ref{extrema}) to the linear order in $\delta$ 
yields that for $\alpha \leq \alpha^{\ast}$ the density is convex, and a local maximum exists at $\delta=0$ in the 
opposite case.          

As is already clear, the particle densities measured simultaneously at distinct points are correlated because the 
structure of the underlying media is unknown. The correlations in the fugacity profile reflect this fact. In 
principle, it is possible to evaluate general multiple point correlation functions for the fugacities using the 
expansion
\begin{equation}
\Erm \prod_{i=1}^n \phi(x_i) = \sum_{k=0}^n v^{n-k} (u-v)^{k} \sum \Erm \prod_{j=1}^k \frac{S(x_{i_j})}{S(1)},
\end{equation}
where the summation is over all subsets of $\lbrace x_i \rbrace$ of cardinality $k$. The calculations, using the same 
techniques as in the proof of the theorem above, are rather tedious for general $n$, and we just mention the result 
for the covariance. For $0 \leq x \leq y \leq 1$, 
\begin{equation}
\Erm \phi(x)\phi(y) - \Erm \phi(x) \Erm \phi(y) = (1-\alpha) (u-v)^2 x (1-y),
\end{equation}
which shows positive correlations.

\section{Motion of a tagged particle}
\label{SecTagged}

We shall first study the motion of a tagged particle in a boundary driven gas that has reached its steady state in 
fixed environment, but eventually we shall consider the case of random media. The convergence of the tagged particle 
processes to one-dimensional diffusions characterized by scale functions and speed measures (as explained for example 
in \cite{RogersWilliamsVol2, RevuzYor}) is not within our reach, but we content ourselves with the analysis of 
expectations of functionals of the particle position integrated up to the time of exit from a set.

\subsection{Travels in fixed environments} 
\label{SecTravelsFinite} 

Let the symmetric environment be fixed, and let $J= \lbrace a+1, a+2, \ldots ,b-1 \rbrace \subseteq \lbrack 
n\rbrack$. The particle system $(X_i)_{i=1}^n$ is assumed to be in the stationary state at $t=0$. Three types of 
paths of the tagged particle corresponding to three experimental situations are considered: In the first case, a 
transition of the particle from $a$ to $a+1$ is observed at $t=0$, and some time later, its first exit from $J$ by a 
jump from $b-1$ to $b$ is registered. It is observed to return to $a$ without reaching $b$ in the second case. 
Finally, the starting position $i\in I$ is known in the third setting, and the time of exit from $J$ by a jump 
through either end is measured.   

The position of the tagged particle will be denoted by $Y_t$ and the exit time from $J$ is $T_{J} = \inf \lbrace 
t\geq 0 : Y_t =a\ \mathrm{or}\ b \rbrace$. Kook and Serfozo \cite{Kook93} have constructed a process such that the 
travel times of the preceding paragraph are stopping times. This process keeps track of the paths that the particles 
are traveling, and labels them accordingly. Let $P^{a\to b}$, $P^{a\hookleftarrow b}$ and $P^{i}$ be the Palm 
measures corresponding the three cases respectively (see \cite{Baccelli03} for Palm calculus). For example,
\begin{equation}
P^{a\to b} (B) = \frac{1}{E N_{a\to b} (0,1 \rbrack } \int_{(0,1 \rbrack} I(Y_{\cdot+t}\in B) 
N_{a\to b}(\ud t),
\end{equation}
where $N_{a\to b}$ counts the particles that, in the stationary state, start the trip from $a$ to $b$ and eventually 
complete it. 

It should be noted that the process of numbers of particles does not determine the motion of the tagged particle 
because it does not tell in which order the particles leave a site. In zero-range processes, it is usually implicitly 
assumed that the particle to leave the site is chosen in uniform random manner from the particles at the site just 
before the transition takes place. The results in this article hold for any rule that assigns a fraction $f_j$, with 
$\sum_{j\geq 1} f_j =1$, of the total "work" $g(X_i)$ to the $j$-th particle of the total of $X_i$ particles present, 
ranked by the order of their arrival to the site $i$. Notice that this extra freedom does not induce biased movement 
of a tagged particle. 

We now proceed to the calculation of expectations a class of functionals of the tagged particle paths under the Palm 
measures $P^{a\to b}$, $P^{a\hookleftarrow b}$, and $P^{i}$. The functionals are of a specific type common in 
physics: They are functions of the particle position integrated up to the time of exit from a set under 
consideration. 

Define the continuum version of the function $s_j$ in formula (\ref{s_j}) as 
$s(x)=1 + \sum_{i=1}^{\lfloor x \rfloor -1} v_i^{-1}$. We shall also need the functions 
\begin{equation}
G_J (x,y) = \frac{\lbrack s(x\wedge y)-s(a) \rbrack \lbrack s(b) - s(x\vee y) \rbrack }{s(b)-s(a)}\, 
I\lbrace x,y \in (a,b) \rbrace
\end{equation}
and measures 
\begin{equation}
m_n (\ud x) = \sum_{i=1}^n \frac{Z'(\phi_i)}{Z(\phi_i)}\, \delta_{i} (\ud x),
\end{equation}
where $\phi_i$ is the fugacity of equation (\ref{fuga_i}), $Z$ is the partition function from equation 
(\ref{statgeneral1}), and $\delta_i$ is the Dirac measure at $i$.

\begin{thm}
For a function $f:{\mathbb Z} \longrightarrow {\mathbb R}$, 
{\setlength\arraycolsep{2pt}
\begin{eqnarray}
E^{a\to b} \int_0^{T_{J}} f(Y_t)\, \ud t &=& \int_{-\infty}^{\infty} f(y) G_{J}(y,y) \, m_n (\ud y) \\
\label{hookleft}
E^{a\hookleftarrow b} \int_0^{T_{J}} f(Y_t)\, \ud t &=& \int_{-\infty}^{\infty} f(y) 
\frac{v_a (s(b)-s(y))^2}{(s(b)-s(a))(s(b)-s(a+1))} \, m_n (\ud y) \\
\label{FixedStartingPoint}
E^{i} \int_0^{T_{J}} f(Y_t)\, \ud t &=& \int_{-\infty}^{\infty} f(y) G_{J}(i,y) \, m_n (\ud y).
\end{eqnarray}}
\end{thm}
\begin{proof}
Let us start with the first assertion. The proof is based on the Little laws (\cite{Whit90,Baccelli03,Serfozo99}) for 
general travel times in queuing networks as expressed by Kook and Serfozo \cite{Kook93}. $T_J (i)$ denotes the time 
the particle spends at $i$ before $T_J$, and $p_i(a\to b)$ is the probability that a particle at $i$ is traveling 
from $a$ to $b$ without return to $a$. Then
\begin{eqnarray}
\lefteqn{E^{a\to b} \int_0^{T_{J}} f(Y_t) \, \ud t = E^{a\to b}  \sum_{i=a+1}^{b-1} f(i) T_J (i)} \nonumber\\
&=& \sum_{i=a+1}^{b-1} f(i) E^{a\to b} T_J (i) = \sum_{i=a+1}^{b-1} f(i) p_i(a\to b) \frac{E X_i}{E N_{a\to b} (0,1 
\rbrack},
\end{eqnarray}
where the last equality holds due to theorem 3.1 in reference \cite{Kook93}. According to corollary 3.4 in the same 
reference, the probability $p_i(a\to b)$ can be expressed as product $\alpha_i(a,b) \alpha_i^{\ast} (a,b)$ of the 
probabilities that a particle at $i$ is absorbed at $b$ before hitting $a$, and that a particle obeying the adjoint 
dynamics (this is why we calculated $L_n^{\ast}$ in the preliminaries section) is absorbed at $a$ before reaching 
$b$. Clearly, the first of these satisfies 
\begin{equation}
\label{alpharec}
\alpha_i (a,b) = p_i \alpha_{i+1}(a,b) + q_i \alpha_{i-1} (a,b)
\end{equation} 
with boundary conditions $\alpha_a (a,b)=0$, $\alpha_b (a,b)=1$, and the shorthand $p_i=v_i /(v_{i-1}+v_i)$, 
$q_i=1-p_i$. By proposition 2, an equation similar to (\ref{alpharec}) holds for the probabilities $\alpha_i^{\ast}$ 
with $p_i^{\ast} = v_i^{\ast}/(v_i^{\ast}+u_i^{\ast}) = v_i \phi_{i+1}/\lbrack (v_{i-1}+v_i)\phi_i \rbrack$ instead 
of $p_i$, and $q_i^{\ast} =1- p_i^{\ast}$ instead of $q_i$. The boundary conditions are replaced by $\alpha_a^{\ast} 
(a,b)=1$, $\alpha_b^{\ast} (a,b)=0$. The solution to this set of equations is 
{\setlength\arraycolsep{2pt}
\begin{eqnarray}
\alpha_i (a,b) &=& \frac{\phi_i -\phi_a}{\phi_b -\phi_a} , \\
\alpha_i^{\ast} (a,b) &=& \frac{\phi_a (\phi_b -\phi_i)}{\phi_i (\phi_b -\phi_a)}.
\end{eqnarray}}
Noting that $E X_i = \phi_i Z'(\phi_i)/Z(\phi_i)$, and $E N_{a\to b}(0,1\rbrack = v_a \phi_a \alpha_{a+1}(a,b)$,  
some simple algebra finishes the proof of the first part. The second statement is shown in the same fashion. Observe 
that $p_i (a\hookleftarrow b) = \alpha_i (b,a) \alpha_i^{\ast}(a,b)$, and $E N_{a\hookleftarrow b} (0,1\rbrack = v_a 
\phi_a \alpha_{a+1}(b,a)$.

The third part of the theorem can be shown to hold as a corollary to the other two. For $a< i < j < b$, 
\begin{eqnarray}
\lefteqn{E^i T_J (j) = E^{i\to b} T_J (j) \frac{v_i\alpha_{i+1}(i,b)}{v_{i-1}\alpha_{i-1}(i,a) +v_i \alpha_{i+1} 
(i,b)}}
\nonumber\\
&&{}+E^{i\hookleftarrow b} T_J (j)\sum_{n\geq 0} \left( \frac{v_{i-1}\alpha_{i-1} (a,i)}{v_{i-1}+ v_i }  +  
\frac{v_{i}\alpha_{i+1} (b,i) }{v_{i-1}+v_i } \right)^n \left( \frac{v_{i-1}\alpha_{i-1} (i,a) }{v_{i-1}+ v_i } +  
\frac{v_{i}\alpha_{i+1} (i,b) }{v_{i-1}+v_i } \right)  \nonumber\\
&&{}\times\sum_{k=1}^n k {{n}\choose{k}} \left(\frac{v_i \alpha_{i+1} (b,i)}{v_{i-1}\alpha_{i-1}(a,i)+ v_i 
\alpha_{i+1} (b,i)} \right)^k \left(\frac{v_{i-1} \alpha_{i-1} (a,i)}{v_{i-1}\alpha_{i-1}(a,i)+ v_i \alpha_{i+1} 
(b,i)} \right)^{n-k}\nonumber\\
&=& E^{i\to b} T_J (j) \frac{v_i\alpha_{i+1}(i,b)}{v_{i-1}\alpha_{i-1}(i,a) +v_i \alpha_{i+1} (i,b)} + 
E^{i\hookleftarrow b} T_J (j) \frac{v_i \alpha_{i+1}(b,i)}{v_{i-1}\alpha_{i-1}(i,a) +v_i \alpha_{i+1} (i,b)}. 
\end{eqnarray}
In the first equality, the time spent in $j$ while traveling from $i$ to either boundary is split into the time in 
$j$ for the trip from $i$ to the right boundary without return, and into the time in $j$ for the paths that do return 
to $i$. The outermost sum in the latter part is over the number of the returns. The inner sum then counts the 
expected number of right excursions. The left excursions do not contribute because $i<j$. The cases with $j \leq i$ 
can be treated similarly. 
Substitution of the formulae for the expected occupation times and simple algebra now finishes the proof. 
\end{proof}

\subsection{Travels in random media} 
\label{SecTravelRandom} 

As an application to the theorem of the preceding section, we now study the behavior of the tagged particle on 
macroscopic spatial intervals. To this end, we define $J_n := \lbrace \lfloor an \rfloor +1, \ldots, \lfloor bn 
\rfloor -1 \rbrace$, where $0\leq a < b \leq 1$. The results of this section are presented informally and the proofs, 
or at least the sketches, are given within the text.

\subsubsection{Weak disorder: $\Erm v_i^{-1}<\infty$}

Let us again deal with the case of finite expected inverse rates first. To get nontrivial limits for the integrals in 
theorem 2 as $n\to \infty$, the motion of the tagged particle must be suitably speeded up. For this purpose, we shall 
consider the processes $(Y_{n^2 t})_{t\geq 0}$. Notice that the exit time $T_J^{(2)}$ of this speeded up process from 
a set $J$ is identical in law with $n^{-2} T_J$.
Let $f:\lbrack 0,1 \rbrack \rightarrow {\mathbb R}$ be continuous. Recall that by the assumption (\ref{uvcond}) on 
the boundary rates, the measures $m_n$ are bounded, with the weights of the Dirac measure given by a uniformly 
continuous function of the fugacity $\phi_n(x)$, and that the functions $S_{n,1} (x) = (1 + \sum_{i=1}^{\lfloor 
(n+1)x \rfloor -1} v_i^{-1})/n$ converge uniformly almost surely. Then, by equation (\ref{FixedStartingPoint}) of 
theorem 2,
{\setlength\arraycolsep{2pt}
\begin{eqnarray}
\label{Ydiffscaling}
& &E^{\lfloor xn \rfloor } \int_0^{T_{J_n}^{(2)}} f\left( \frac{1}{n} Y_{n^2 t} \right)\, \ud  t 
= n^{-2} \int_{-\infty}^{\infty} f(y/n) G_{J_n}(\lfloor xn \rfloor ,y) \, m_n (\ud y) \nonumber\\
& & \qquad\qquad\qquad\qquad \longrightarrow \frac{1}{u-v} \int_a^b f(y) \frac{(y\wedge x- a)(b- y\vee x)}{b-a} \, 
m(\ud y)
\end{eqnarray}}
for almost all environments as $n\to \infty$, where $m(\ud y) = \lbrack Z'(uy +v(1-y)) / Z (uy + v(1-y)) \rbrack \, 
\ud y$, because the integrand of the expression after the first equality, written in terms of the macroscopic 
coordinate $y\in (a,b)$, converges uniformly. The same scaling leads to a nontrivial almost sure limit also for the 
conditioned motion from $\lfloor an \rfloor$ to $\lfloor bn \rfloor$, i.e. the first case of the theorem. The only 
change is that $x$ is replaced by $y$ on the right hand side of the equation (\ref{Ydiffscaling}). For the paths that 
are forced to return to their starting point, however, no such limit exists because the expectations with respect to 
$P^{\lfloor an \rfloor \hookleftarrow \lfloor bn \rfloor}$ depend strongly on the transition rate $v_{\lfloor an 
\rfloor}$ at the boundary, cf.\ equation (\ref{hookleft}).

\begin{example}
The travel time of a particle through the whole system in a gas with the attractive interaction $g(k)=I\lbrace k>0 
\rbrace$ satisfies
{\setlength\arraycolsep{2pt}
\begin{eqnarray}
\label{TravelTimeSelfAve}
& & n^{-2} E^{0\to n+1} T_{[n]} \longrightarrow \int_0^1 \frac{y(1-y)}{1-uy-v(1-y)}\, \ud y \nonumber\\
& &\ \ \ = \frac{2-u-v}{2(u-v)^2} + \frac{(1-u)(1-v)}{(u-v)^3} \log \frac{1-u}{1-v}. 
\end{eqnarray}}
Taking $v=0$ and sending $u\to 0$ yields the value $1/6$, which is of course the same as one would get by calculating 
the expected travel time in a gas of noninteracting particles, that is $g(k)=k$. Notice that, somewhat surprisingly, 
the travel time (\ref{TravelTimeSelfAve}) is finite even for one of the boundary rates tending to 1 (i.e. the radius 
of convergence of $Z$).
\end{example}

\subsubsection{Strong disorder: $\Erm v_i^{-1}=\infty$}

To prove that the integrals of the functions of the tagged particle position converge for $v_i^{-1}$ in domain of 
attraction of $\alpha$-stable variables with $\alpha \in (0,1)$, the functional limit theorems of Skorohod 
\cite{Skorohod57} are needed. Of course, there is no hope of getting statements concerning almost all environments 
for nontrivial paths of the tagged particle as in the case of weak disorder, but the aim is to determine the 
distributions of the travel functionals. 

It turns out that theorem 2.7 of reference \cite{Skorohod57} ideally suits our purposes: The theorem concerns the 
distributional convergence of continuous (in the Skorohod topology) functionals of sums of independent and 
identically distributed random variables under the assumption of pointwise convergence. Thus only continuity of the 
integrals over the macroscopic intervals $(a,b)$ must be proven. But this follows from the almost sure continuity of 
the function that maps the paths of $S_{n,\alpha}$ to the integrand (all this written again in the macroscopic 
coordinates) in the Skorohod topology, and the fact that the Skorohod topology is finer than that given by the metric 
$\int_0^1 \vert \phi(x) -\psi(x) \vert\, \ud x$ \cite{Billingsley68}. So the functionals of the tagged particle 
process speeded up by $n^{1+1/\alpha}$ converge as
\begin{equation}
\label{Yalphascaling}
E^{\lfloor xn \rfloor } \int_0^{T_{J_n}^{(1+1/\alpha)}}\!\!\! f\left( \frac{1}{n} Y_{n^{1+1/\alpha} t} \right)\, \ud 
t \longrightarrow 
\int_a^b f(y) \tilde{G}_J (x,y) \, \tilde{m} (\ud y)
\end{equation}
in distribution, with $\tilde{G}_J (x,y) = ( S(x\wedge y)- S(a) )( S(b)- S(y\vee x) )/( S(b)-S(a))$ and $\tilde{m} 
(dy)= \lbrack Z'(\phi(y))/Z(\phi(y)) \rbrack\, \ud y$. Here $S$ is again the $\alpha$-stable subordinator and 
$\phi(x) = u S(x)/S(1) + v (1- S(x)/S(1))$. As in the case of weak disorder, the convergence holds for asymmetric 
paths through intervals as well. The "propagator" must then be substituted with $\tilde{G}_J (y,y) = ( S(y)- S(a) )( 
S(b)- S(y) )/( S(b)-S(a))$.

The expectations over the disorder (with respect to $\Prm$) of the travel time integrals do not exist even in finite 
systems for $\alpha \leq 1/2$. To see this, notice that the coefficients $Z'(\phi_i)/Z(\phi_i)$ are bounded in the 
sums to be evaluated, but the expectations of the propagators $G_J$ in theorem 2 exist for $\alpha > 1/2$ only: In 
the simplest case of the set $J$ consisting of just one site, 
{\setlength\arraycolsep{2pt}
\begin{eqnarray}
\Erm \frac{v_{i-1}^{-1}v_i^{-1}}{v_{i-1}^{-1}+v_i^{-1}} &=& \Erm \frac{1}{v_{i-1}+v_i} 
\sim \int_1^{\infty} \int_1^{\infty} \frac{x^{-\alpha-1} y^{-\alpha-1}}{x^{-1}+y^{-1}} \, \ud x \ud y \nonumber\\
&=& \int_1^{\infty} (x(1+x))^{-\alpha} \sum_{k\geq 0} \frac{(\alpha)_{k\uparrow}}{k!(\alpha+k)} \left( 
\frac{x}{1+x}\right)^k \, \ud x,  \nonumber
\end{eqnarray}}
with $(\alpha)_{k\uparrow}= \alpha(\alpha+1)\cdots (\alpha+k-1)$. The sum being uniformly bounded in $x$, the result 
follows. The lack of convergence of expectations for $\alpha \leq 1/2$ is inherited by larger finite systems and, 
eventually, by the continuum limit. Indeed, using the same machinery as for calculating the moments and the two-point 
correlations of the fugacity in section \ref{SecSrm}, we get
\begin{equation}
\Erm \, \tilde{G}_J (x,y) = 
\alpha \Gamma (2-1/\alpha) (b-a)^{1/\alpha -1} \Phi(1)^{1/\alpha} \frac{(x-a)(b-y)}{b-a}, 
\end{equation}
for $x < y$. Here the gamma function diverges as $\alpha \to 1/2$ from above.

\begin{example}
Let us consider our canonical interacting system with $g(k)=I\lbrace k>0 \rbrace$ and, in particular, the expected 
time it takes for a particle to travel from the left to the right reservoir. Average over the media is also taken.  
We assume that $\alpha \in (1/2,1)$. Then
\begin{displaymath}
\lim_{n\to\infty} n^{-(1+1/\alpha)} \Erm E^{0\to n+1} T_{[n]}
= \frac{1}{u-v} \,\Erm \int_0^1 \left( S(1)-S(y) \right) \sum_{k\geq 1} z^k \left(\frac{S(y)}{S(1)} \right)^{k} \, 
\ud y,
\end{displaymath}
where $z=(u-v)/(1-v)$. As in the proof of theorem 1, we apply Fa\`a di Bruno formula in the evaluation the moments of 
$S(y)/S(1)$. In this case, the resulting summations can be carried out explicitly because, instead of a sum over 
partitions of a fixed number $k$, we have sums over partitions of numbers whose magnitude is determined by the 
exponential weight $z^k$, and the order of these summations can be interchanged. As a result, the expected travel 
time equals 
\begin{displaymath} 
\frac{1}{u-v} \int_0^1 (1-y) \int_0^{\infty} \alpha r^{\alpha-2} \Phi (1) e^{-r^{\alpha}\Phi (1)} z 
\frac{\ud}{\ud z} e^{-yr^{\alpha}\Phi (1) ((1-z)^{\alpha}-1)} \, \ud r \ud y.
\end{displaymath}
Performing the derivative and the integrations finally yields the result
\begin{equation}
\label{AlphaTravelTime}
\frac{\alpha^3\Phi(1)^{1/\alpha} \Gamma(2-1/\alpha) \lbrack (1-v)(1-u) \rbrack^{\alpha-1} }{(1-\alpha^2)  \lbrack 
(1-v)^{\alpha} - (1-u)^{\alpha} \rbrack^2} \Big\lbrace 2-u-v - \frac{\alpha (u-v) \lbrack (1-v)^{\alpha} + 
(1-u)^{\alpha} \rbrack}{ (1-v)^{\alpha} - (1-u)^{\alpha} } \Big\rbrace
\end{equation}
for it. Taking $\alpha \to 1$ leads, up to a model dependent factor $\Phi(1)$, to equation (\ref{TravelTimeSelfAve}) 
of the previous example. Nevertheless, the travel time (\ref{AlphaTravelTime}) diverges for any $\alpha\in (1/2, 1)$ 
as one of the boundary rates approaches the critical value 1.   
\end{example}

\section{Discussion} 
\label{SecDisc}

In the preceding sections, we have shown that the stationary state of a boundary driven zero-range process in 
symmetric random media, with divergent expected inverse jump rates, is in the large system limit described by a 
normalized $\alpha$-stable subordinator. It was also argued that, in this limit, the physical quantities related to 
the motion of a suitably speeded up tagged particle on an interval $J$ can be calculated using the formula 
$\tilde{G}_J (x,y)=(S(x\wedge y)-S(a))(S(b)-S(x\vee y))/(S(b)-S(a))$ for the propagator or "Green's function", and 
the measure $\tilde{m} (\ud y) = Z'(\phi(y))/Z(\phi(y)) \, \ud y$. In mathematics, such function $S$ is known as the 
scale function and the measure $\tilde{m}$ as the speed measure \cite{RevuzYor, RogersWilliamsVol2} of a diffusion. 
In fact, they {\it define} a one-dimensional diffusion even when the stochastic differential equation description is 
purely formal, e.g.\ with discontinuous coefficients as in our case. In practice, the scale function and the speed 
measure tell how the spatial and time axes must be stretched for the motion of the particle to be indistinguishable 
from a Brownian motion (or vice versa, how the motion of the particle can be constructed from a Brownian motion).

Kawazu and Kesten \cite{Kawazu84} have proved an invariance principle for a random walk in symmetric random media. 
The scale function and the speed  measure of the motion on large scales coincide with $S$ and $\tilde{m}$ given in 
this article. The speed measure is just the Lebesgue measure in this case. In the article of Kawazu and Kesten, the 
discrete random walk was on the infinite set ${\mathbb Z}$, so the only difference to our setting in the 
noninteracting case is that our particles are absorbed at $0$ and $n+1$. For a convergence proof for zero-range 
processes on tori with fixed rates, see \cite{Siri98}. It would be of interest to find a proof that the speeded up 
motion of a tagged particle in a zero-range process in symmetric random media converges to a diffusion defined by the 
scale function and speed measure given in this article. 

\bigskip

\begin{acknowledgments}
The author would like to thank Stefan Geiss, Pekka Kek\"al\"ainen, and Juha Merikoski for discussions. This work has 
been supported by the Centre of Excellence program of the Academy of Finland. 

\end{acknowledgments}

\break

\vfill

\eject

\end{document}